\documentclass[conference]{IEEEtran}
\IEEEoverridecommandlockouts
\usepackage{cite}
\usepackage{amsmath,amssymb,amsfonts,amsthm}
\usepackage{graphicx}
\usepackage{textcomp}
\usepackage{xcolor}
\usepackage{algorithm}
\usepackage{algpseudocode}
\usepackage[hidelinks]{hyperref}
\def\BibTeX{{\rm B\kern-.05em{\sc i\kern-.025em b}\kern-.08em
    T\kern-.1667em\lower.7ex\hbox{E}\kern-.125emX}}

\newcommand{\RR}{\mathbb{R}}
\newcommand{\ii}{\mathrm{i}}
\newcommand{\dd}{\mathrm{d}}
\newcommand{\Real}{\Re}

\newtheorem{thm}{Theorem}
\newtheorem{defn}{Definition}

\begin{document}
\title{Consistent Spectrogram Separation\\ from Nonstationary Mixture}

\author{\IEEEauthorblockN{Marina Krémé}
\IEEEauthorblockA{\textit{Mines Saint-Etienne, Institut Henri Fayol} \\
Saint-Etienne, France \\
marina.kreme@emse.fr}
\and
\IEEEauthorblockN{Adrien Meynard}
\IEEEauthorblockA{\textit{Univ Lyon, ENS de Lyon, CNRS, Laboratoire de Physique} \\
Lyon, France \\
adrien.meynard@ens-lyon.fr}
}

\maketitle

\begin{abstract}
We present a spectrogram separation method tailored for mixtures comprising two nonstationary components. By exploiting the unique characteristics of their time-frequency representations, we propose an inverse problem formulation to estimate the spectrograms of the components. We then introduce an alternating optimization algorithm that ensures the consistency of the estimated spectrograms. The efficacy of the algorithm is evaluated through testing on synthetic mixtures and is applied to a bioacoustic signal. 
\end{abstract}

\begin{IEEEkeywords}
time-frequency, inverse problem, source separation, consistency,  single-channel
\end{IEEEkeywords}

\section{Introduction}
Spectrogram separation constitutes an essential first step in many applications before performing the desired task. For example, in applications such as sound event detection and localization, speech enhancement, or single-channel audio source separation \cite{Reddy2007, Rafii2011ASM, Hershey2016}, the proposed approaches take a mixture spectrogram as input and generate a mask to apply to the spectrogram for each source being separated. 
Although the primary focus is on audio source separation, the first step is spectrogram separation. We thus define \textit{spectrogram separation as the task of decomposing a spectrogram into several spectrograms based on the different patterns that constitute it}, similar to the texture image decomposition \cite{aujol2006structure}. In this paper, we focus on the spectrogram separation for mixtures of two components. This can be seen as the initial stage in single-channel source separation~\cite{Benaroya2006}, which is 
a major challenge, representing the extreme scenario of underdetermined blind source separation (when the number of sources exceeds the number of observations)~\cite{Hopgood1999}. 
Various methods tackle this challenge leveraging source signal characteristics and application requirements: 
Probabilistic models \cite{Kristjansson, Juang1991, Virtanen2006}, spectral decomposition-based methods, Computational Auditory Scene Analysis (CASA)-based methods~\cite{Drake2009}, and the deep neural network (DNN) approaches \cite{Grais2014}.


Probabilistic models, including Gaussian mixture models (GMM)\cite{Kristjansson, Jang2004}, hidden Markov models (HMM)~\cite{Juang1991,Gales2007}, and factorial HMM~\cite{Virtanen2006, Roweis2000OneMS}, are widely utilized for speaker separation tasks. While GMM and HMM assume constant source energy during separation, limiting real-time performance, factorial HMM models mitigate this constraint but increase computational complexity.
An alternative approach involves spectral decomposition-based methods such as Independent Subspace Analysis (ISA)~\cite{Casey2000SeparationOM} and Nonnegative Matrix Factorization (NMF) \cite{Schmidt2006, Grais2011, Fevotte2018}. 
Independent Subspace Analysis (ISA) is an extension of Independent Component Analysis (ICA)~\cite{Comon2010handbook} for single-channel source separation, aiming to decompose the time-frequency space of a mixed signal into independent source subspaces using Short-Time Fourier Transform (STFT), which often results in cross-spectral terms due to harmonic phenomena and overlapping windows between consecutive time frames. Spectral decomposition methods use the principle of NMF \cite{Grais2011}. However, the energy of real sources can be negative or positive. 
CASA\cite{Drake2009} aims to separate the mixture of sound sources like human ears do. These techniques, however, have difficulty separating instruments with similar tonal characteristics into distinct streams. Recently, DNN-based approaches have shown effective separation in the TF domain for speech signals. Nevertheless, these methods demand an extensive database for training, which may not always be readily accessible.

Unlike these methods, which work on the reconstruction of the sources, we focus on the reconstruction of the  sources spectrograms from the mixture, due to their pivotal role in determining sources frequency content. 
Accurate spectrogram estimation lays the groundwork for subsequent phase estimation, essential for signal reconstruction. To achieve this, we impose the \textit{consistency} constraint \cite{Yatabe2020} on the estimated spectrograms in our algorithm: indeed, since the STFT is not surjective, estimated spectrograms may not correspond to those of a signal. The consistency constraint ensures this.  In Section \ref{sec:modeldef}, we describe and introduce our model. We then present our method in Section \ref{sec:method}. Section \ref{se:results} presents the numerical experiments and results. Section \ref{sec:concl} pertains to the conclusion.

\section{Models and Definitions}
\label{sec:modeldef}
Consider a mixture $z$ consisting of two distinct nonstationary signals $x$ and $y$:
\begin{equation}
z(t) = x(t) + y(t) ,
\end{equation}
where $x$ is a \emph{bumps signal} and $y$ is a \emph{multicomponent AM-FM signal} (see Definitions~\ref{def:x} and~\ref{def:y}).

\begin{defn}[Bumps signal]
\label{def:x}
A bumps signal $x$ is composed of a series of $K$ distinct impulses, called bumps. It is written:
\begin{equation}
\label{eq:model.x}
x(t) = \sum_{k=1}^K \varphi(t-t_k),
\end{equation}
where $t_k$ are the impulse times, and $\varphi$ defines an even, bounded, function, called bump, localized near the origin. 

We denote by $\Delta_\varphi$ the constant such that the support $\varphi$ is $[-\Delta_\varphi,\Delta_\varphi]$. Furthermore, the impulses times $t_k$ are assumed to be distinct: the smallest difference between two consecutive pulses is denoted by $\theta_\varphi$, i.e.,
\[
\theta_\varphi =\min_{k,k'} |t_k-t_{k'}|.
\]
\end{defn}

\begin{defn}[Multicomponent AM-FM signal]
\label{def:y}
A multicomponent AM-FM signal $y$ is defined as: 
\begin{equation}
\label{eq:model.y}
y(t) = \sum_{\ell=1}^L A_\ell(t)\sin(2\pi\phi_\ell(t)),
\end{equation}
for $L \in \mathbb{N}$, $A_\ell$ and $\phi_\ell$ are respectively instantaneous amplitude and phase of mode $\ell$ satisfying: $A_\ell(t)>0$, $\phi'_\ell(t)>0$ and $\phi'_{\ell+1}(t)>\phi'_{\ell}(t)$ for all $t$, where $\phi'_\ell$ is referred to as the instantaneous frequency.
\end{defn}
The bumps component is localized around the impulse times $t_k$, and vanishes elsewhere, characterizing brief and irregular events. The signal $x$ is therefore well localized in time. In contrast, signal $y$ is locally periodic, making it well-localized in the frequency domain. Hence, time-frequency representations appear to be suited to isolate the contributions of each component separately.

\begin{defn}[Short-Time Fourier Transform (STFT)]
Let $x\in L^2(\RR)$ the analyzed signal, $g\in L^2(\RR)$ the analysis window. The STFT $T_x$  of $x$ is defined by: 
\begin{equation}
\label{eq:stft}
T_x(\nu,\tau) = \int_{\RR} x(t) g(t-\tau) e^{-\ii 2\pi\nu t} \dd t .
\end{equation}
The spectrogram, denoted by $S_x$, is the squared modulus of the STFT, that is,
\begin{equation}
S_x(\nu,\tau) = \left| T_x(\nu,\tau) \right|^2 .
\end{equation}
\end{defn}

\section{Method}
\label{sec:method}
\subsection{Asymptotic results}
\label{sse:asymptotic}

In this section, we provide the asymptotic behaviors of the spectrograms of signals as defined in Definitions~\eqref{def:x} and~\eqref{def:y}.

\subsubsection{Spectrogram of the bumps signal}
\begin{thm}\label{th:approx.dSx}
Let $x$ be a signal defined by ~\eqref{eq:model.x}. Choose a compactly supported, differentiable analysis window $g$, with $\mathrm{supp}(g)=[-\Delta_g,\Delta_g]$, so that it satisfies
\begin{equation}
\Delta_\varphi<\Delta_g<\dfrac{\theta_\varphi}{2}-\Delta_\varphi .
\label{eq:condition.Delta}
\end{equation}
Then, the partial derivative of its spectrogram with respect to frequency, denoted by $\partial_\nu S_x$, satisfies the following equation:
\begin{align}
\nonumber
\partial_\nu S_x(\nu,\tau)\! =\! 2\Real\Bigg(\!\sum_{k=1}^K& \left( g(t_k-\tau)\hat\varphi(\nu) + \epsilon_1(\nu,\tau-t_k) \right) \\
&\hspace{-24pt}\times\left( g(t_k-\tau) \overline{\hat\varphi'(\nu)} + \epsilon_2(\nu,\tau-t_k) \right)\Bigg)
\label{eq:dSx}
\end{align}
\vspace*{-0.7cm}
where
\begin{align*}
|\epsilon_1(\nu,\tau)|\leq & \Delta^2_\varphi\|\varphi\|_\infty\|g'\|_\infty, \\
|\epsilon_2(\nu,\tau)|\leq & \frac43\pi \Delta^3_\varphi\|\varphi\|_\infty\|g'\|_\infty,
\end{align*}
and $\Real(.)$  denotes the real part.
\end{thm}
\begin{proof}\renewcommand{\qedsymbol}{}
See Appendix.
\end{proof}
Theorem~\ref{th:approx.dSx} shows that the partial derivative of the spectrogram of $x$ with respect to frequency is controlled by the spread of the bump $\varphi$ and the smoothness of the window function $g$, among others.
Inequality~\eqref{eq:condition.Delta} specifies that the analysis window can encompass each bump separately from neighboring bumps. Then, the more localized the bump $\varphi$, the smaller the error terms $\epsilon_1$ and $\epsilon_2$. For the sake of simplicity, we assume these terms are negligible and focus on the remaining cross-terms $|g(t_k-\tau)|^2\hat\varphi(\nu) \overline{\hat\varphi'(\nu)}$. They are nonzero in the vicinity of $t_k$ only. However, since $\varphi$ is localized in time, the uncertainty principle, stipulates that $\hat\varphi$ is spread in frequency. Consequently, its derivative $\hat\varphi'$ remains small over the frequency band considered by the STFT.

\subsubsection{Spectrogram of the multicomponent AM-FM signal}
The properties of the STFT of AM-FM signals have been well-studied for decades. As a matter of fact, the quest for the ideal time-frequency representation has motivated the construction of adaptive methods such as the empirical mode decomposition~\cite{Huang1998empirical}, the reassignment or the synchrosqueezing transform (see~\cite{Auger2013time} for a review). The ideal time-frequency representation (ITF) of a multicomponent AM-FM signal is written as:
\begin{equation}
\label{eq:itf}
\mathrm{ITF}_x(\nu,\tau) = \sum_{\ell=1}^L  A_\ell(\tau) \delta( \nu-\phi'_\ell(\tau)) 
\end{equation}
The instantaneous spectral content of $x$ at time $\tau$ is precisely localized at the instantaneous frequencies. The contribution of each component is weighed by its instantaneous amplitude.

Sparsity is a key property of such a representation. As shown by Equation~\eqref{eq:itf}, the ITF is a sparse representation in the time-frequency plane. This observation has led to the construction of inverse problems whose solution is sparse and approximates the ITF~\cite{Kowalski2018convex}.

\subsubsection{Spectrogram of the mixture}
The purpose of this work is to separate the contributions of each component to the spectrogram of the observations. Since
\[
S_z = S_x + S_y + 2\Real(T_x \overline{T_y}),
\]
we will neglect the interactions between the STFT of $x$ and $y$, and assume that the cross term is negligible. This assumption,  based on the models in Definitions~\ref{def:x} and~\ref{def:y}, is numerically justified in Section~\ref{se:results}---a mathematical proof is part of our upcoming work.

\subsection{Inverse Problem}
\subsubsection{Problem Statement}
Given the approximated behaviors of the spectrograms $S_x$ and $S_y$ provided in Section~\ref{sse:asymptotic}, we estimate them by solving:
\begin{equation}
\hat S_x, \hat S_y =\!\arg\!\min_{S_x,S_y} \!\| S_z - \!(S_x+S_y) \|_2^2 + \!\lambda \|\partial_\nu S_x \|_2^2 + \mu \| S_y \|_1 
\label{eq:problem-formulation}
\end{equation}

The first term, called data fidelity, quantifies the difference between the observed spectrogram and the sum of the approximated spectrograms. It ensures that the solution accurately reconstructs or fits the observed data. The second term introduces the regularization for $S_x$. It promotes the smoothness of $S_x$ as described in Section~\ref{sse:asymptotic}. The third term encourages sparsity in $S_y$, implying that most elements of $S_y$ should be zero. The parameters $\lambda$ and $\mu$ are tuning parameters that balance the importance of fitting the data, controlling the smoothness of $S_x$, and promoting sparsity in $S_y$.

We optimize Problem \eqref{eq:problem-formulation} using an alternating minimization strategy, alternating updates between $S_x$ and $S_y$. The initial estimate $S_y^{(0)}$ serves as the starting point for this iterative process. At iteration $k$, the goal is to find $  S_x^{(k)}$ that minimizes the following cost function:
\begin{equation}
  S_x^{(k)}= \! \arg\min_{S_x} \| S_z -S_x - S_y^{(k-1)} \|_2^2 + \lambda \|\partial_\nu S_x \|_2^2.
  \label{eq:pb.Sx}
\end{equation} 
In the next step, the cost function is updated with the current estimate of  $S_x^{(k)}$ and the algorithm optimizes over $S_y$:
\begin{equation}
   S_y^{(k)}=\arg\min_{S_y} \| S_z -S_x^{(k)}-S_y \|_2^2 + \mu \| S_y \|_1
   \label{eq:pb.Sy}
\end{equation}
For practical applications and the sake of clarity, we consider the following the discrete version of STFT, while keeping the same notations. Due to the quadratic nature of the cost function in \eqref{eq:pb.Sx}, canceling out its gradient leads to:  

\begin{equation}
 S_x^{(k)} =  \left(I+ \lambda B^TB \right)^{-1}\left( S_z -  S_y^{(k-1)} \right),
\label{eq:Sx.update}
\end{equation}
with $B$ the matrix obtained after discretization of $\partial_\nu S_x$:
$$B= \begin{pmatrix}
-1 & 1 & 0 & \cdots & 0 \\
0 & -1 & 1 & \cdots & 0 \\
0 & 0 & -1 & \ddots & \vdots \\
\vdots & \vdots & \ddots & \ddots & 1 \\
0 & 0 & \cdots & 0 & -1 \\    
\end{pmatrix}$$
We solve \eqref{eq:pb.Sy} using the FISTA algorithm~\cite{Beck2009}, which is well-suited for tackling such inverse problems. For more details on the iterations of the algorithm, we refer the reader to~\cite{Beck2009}. 

\subsubsection{Spectrogram consistency} 
In the discrete case as mentioned above, a complex valued-matrix $X$  is said to be \textit{consistent} if $\Pi(X)=X$, where $\Pi = \mathrm{STFT} \circ  \mathrm{STFT}^{-1} $. $\Pi$ is a projector on the subspace of consistent matrices. Since we estimate the spectrograms, adhering to the \textit{consistency} constraint is imperative. Therefore, at each stage of estimating $S_x$ and $S_y$, we incorporate the phases of the mixture $S_z$ and project onto $\Pi$.

\subsubsection{Estimation Algorithm}

Algorithm~\ref{alg:algo} presents the proposed alternate estimation algorithm. Initialization is set so that $S_y^{(0)}=0$. The algorithm stops when it reaches the maximum number of iterations $K$ or when the relative change in the cost function~\eqref{eq:problem-formulation}, denoted by $\rho^{(k)}$, falls below the threshold $\Theta$.

\begin{algorithm}
\caption{Consistent Spectrogram Separation}
\label{alg:algo}
\begin{algorithmic}
\Require $\lambda$, $\mu$ \Comment{Hyperparameters}
\Require $K$, $\Theta$ \Comment{Algorithm parameters}
\Require $T_z$ and $S_z$ \Comment{Input STFT and spectrogram}
\State $S_y^{(0)} \gets 0$ \Comment{Initialization}
\State $k \gets 1$
\While{$k \leq K$ and $\rho^{(k)}\geq \Theta$}
\State $S_x^{(k)} \gets$ Output of \eqref{eq:Sx.update}  \Comment{$S_x$ update}
\State $T_x^{(k)} \gets \sqrt{S_x^{(k)}} \arg(T_z)$
\State $S_x^{(k)} \gets |\Pi(T_x^{(k)})|^2$ \Comment{Consistency of $S_x$}
\State $S_y^{(k)} \gets$ FISTA on \eqref{eq:pb.Sy} \Comment{$S_y$ update}
\State $T_y^{(k)} \gets \sqrt{S_y^{(k)}} \arg(T_z)$
\State $S_y^{(k)} \gets |\Pi(T_y^{(k)})|^2$ \Comment{Consistency of $S_y$}
\State $k \gets k+1$
\EndWhile
\end{algorithmic}
\end{algorithm}

\section{Numerical experiments}
\label{se:results}

We implement the proposed algorithm to both a synthetic mixture and a real audio signal. For the sake of reproducibility, the codes produced for this article are available on GitHub\footnote{\url{https://github.com/AdMeynard/SpectrogramSeparation}}.

\subsection{Application to a synthetic signal}

We generate a 1-second mixture, sampled at $F_\mathrm{s}=2^{14}$~Hz (i.e., $2^{14}$ samples). The bumps signal is composed of about 20 randomly spaced bumps, while ensuring that $\theta_\varphi\geq 35$~ms. The bumps have a Hann window shape, where $\Delta_\varphi=0.55$~ms (i$9$ samples). The AM-FM signal has a single component, with a constant amplitude, and an instantaneous frequency defined by
\[
\phi'_1(t) = f_0\left( 1+\frac{t}{2}\sin(2\pi t) \right),
\]
where $f_0=1.5$~kHz is the central frequency. Figure~\ref{fig:Sz} shows the resulting mixture spectrogram. The analysis window is a Hann window of length $2\Delta_g=31.5$~ms. This satisfies Inequality~\eqref{eq:condition.Delta}, ensuring non-interference between the bumps in the spectrogram of the bumps signal. 

\begin{figure}
\centering
\includegraphics[width=\columnwidth]{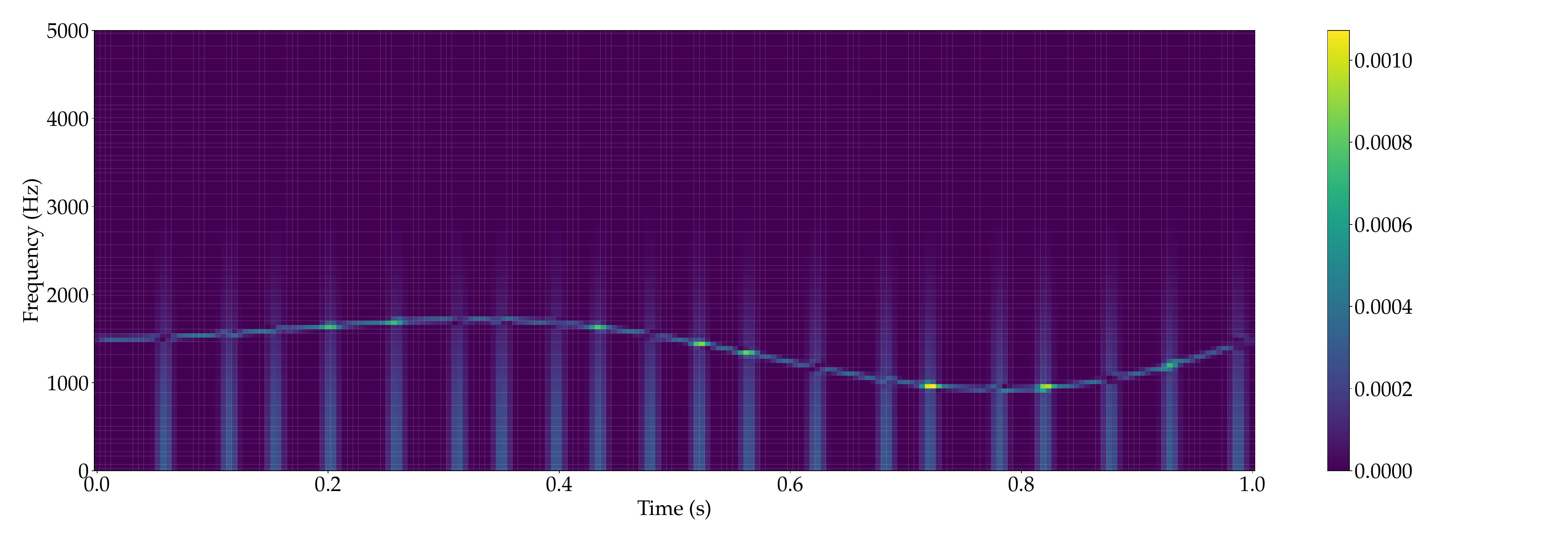}
\caption{Spectrogram of the synthetic mixture.}
\label{fig:Sz}
\end{figure}

We run the proposed algorithm with $\Theta=0.1\%$. We observe that the algorithm converges in $52$ iterations. Figure~\ref{fig:Sx.Sy} shows the estimated spectrograms of $x$ and $y$ in the left panels. 
We note that the two components are indeed separated in the time-frequency plane. The intrinsic characteristics of the spectrograms are well preserved, with smoothness regarding frequency for the first component, and sparsity for the second. However, challenges arise in regions of overlap, making component separation more challenging and leading to some distortions. 
The spectrograms estimated using NMF are shown in the right panels of Figure~\ref{fig:Sx.Sy}.
While NMF accurately captures the vertical patterns of the first component, the estimated second component exhibits a spectrogram resembling a discontinuous horizontal line, diverging significantly from the expected pattern. The additional degree of freedom offered by our method allows for a more accurate decomposition of the spectrogram. Indeed, the norm of residuals $\| S_z - S_x^{(k)}-S_y^{(k)} \|_2$ after convergence is $7.8\times 10^{-5}$ with our algorithm while it is $1.4\times 10^{-2}$ with NMF.
 

\begin{figure}
    \centering
    \includegraphics[width=.49\columnwidth]{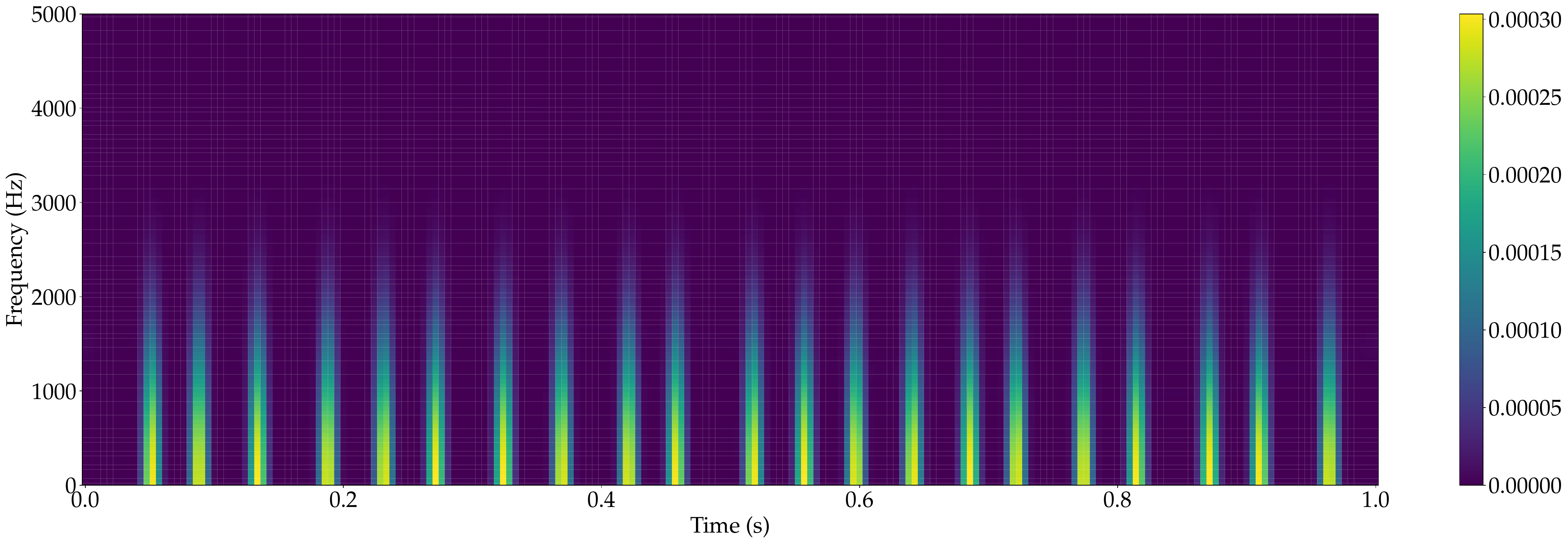}
    \includegraphics[width=.49\columnwidth]{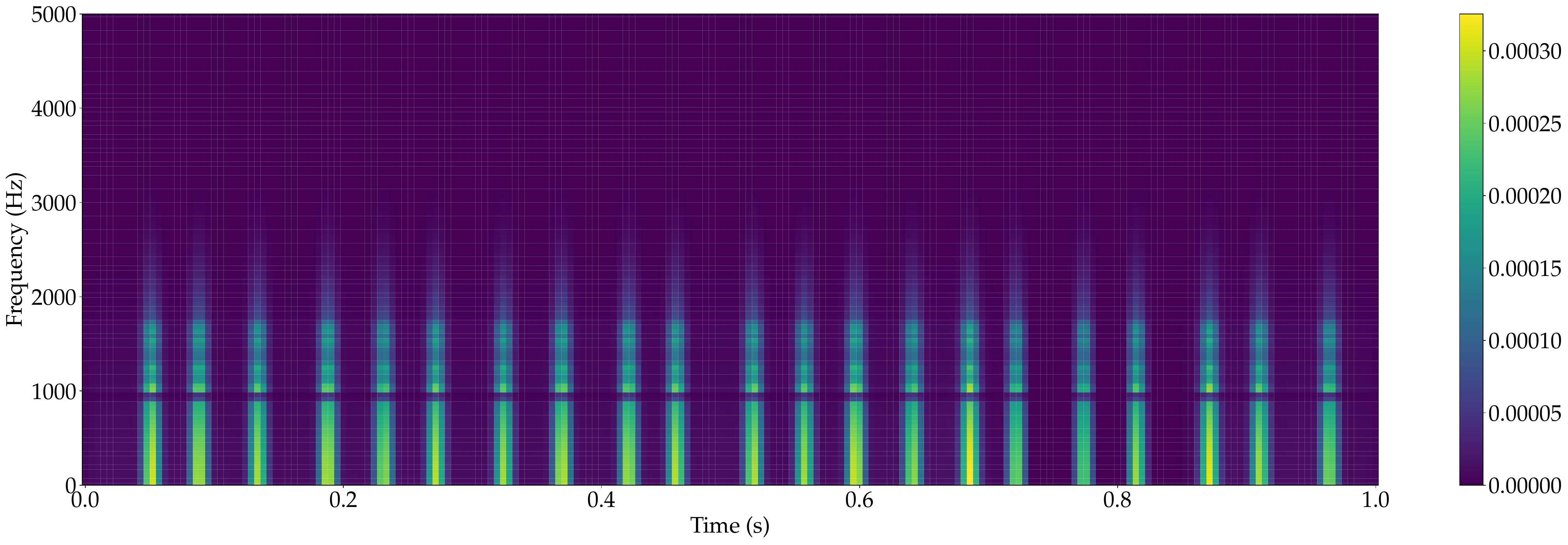}
    \includegraphics[width=.49\columnwidth]{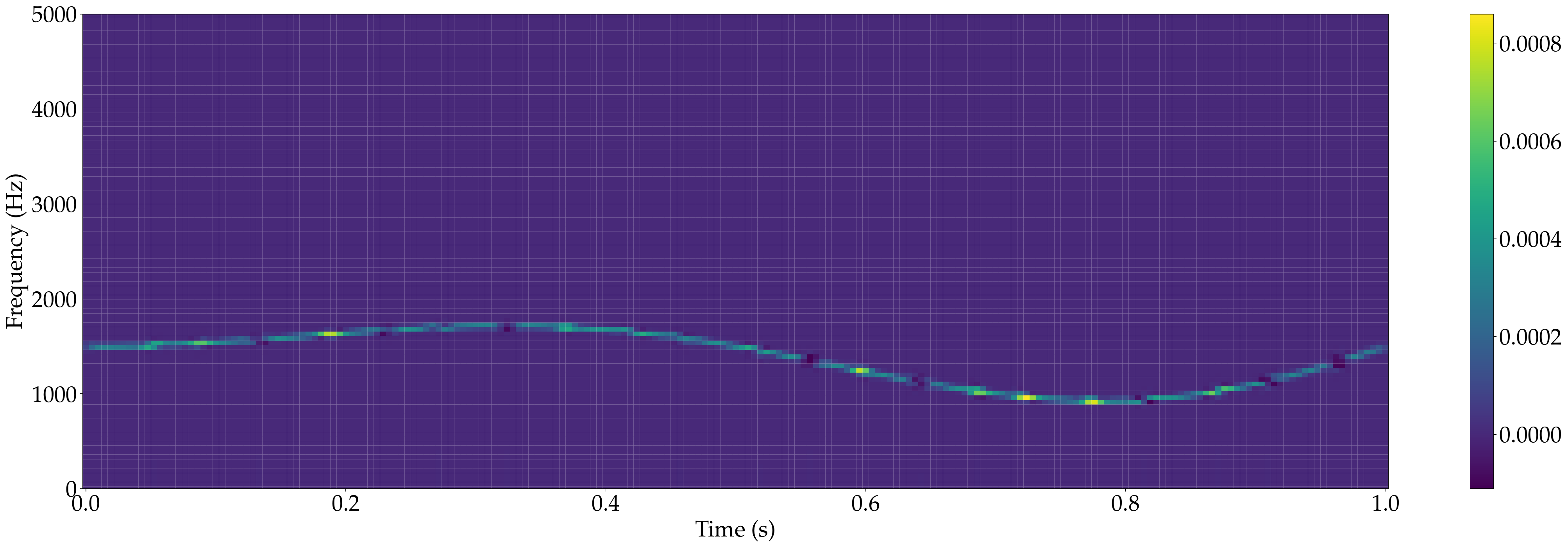}
    \includegraphics[width=.49\columnwidth]{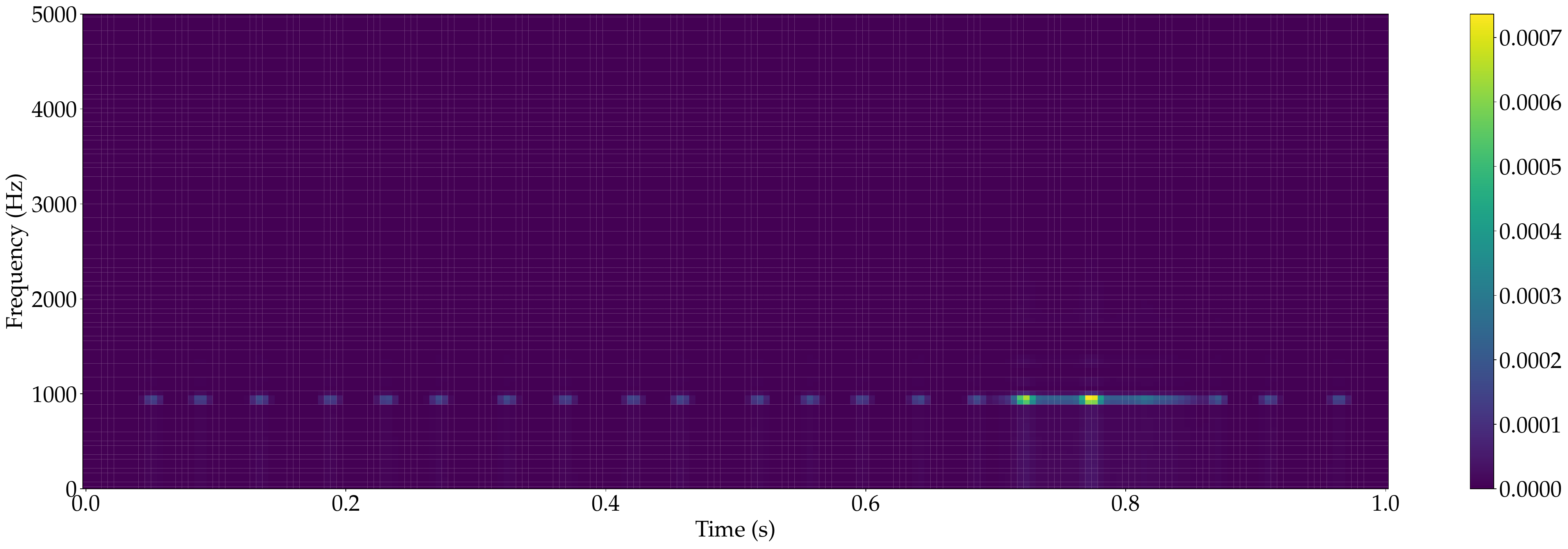}
    \caption{Estimated spectrograms of the bumps signal (top) and the AM-FM signal (bottom). The spectrograms on the left are obtained by our method, those on the right by NMF.}
    \label{fig:Sx.Sy}
\end{figure}

\subsection{Application to an audio signal}

We analyze a $3.5$-second recording of a vocalizing Atlantic spotted dolphin. It is sourced from the Watkins Marine Mammal Sound Database~\cite{Sayigh2016watkins}. The sound the dolphin produces is made up of \emph{whistles} and \emph{clicks}. The bumps signal models the clicks, whereas the AM-FM component models the whistles. The spectrogram of the sound is shown in Figure~\ref{fig:Sz.audio}. The contributions of the two components are visible in the time-frequency plane.

\begin{figure}
    \centering
    \includegraphics[width=\columnwidth]{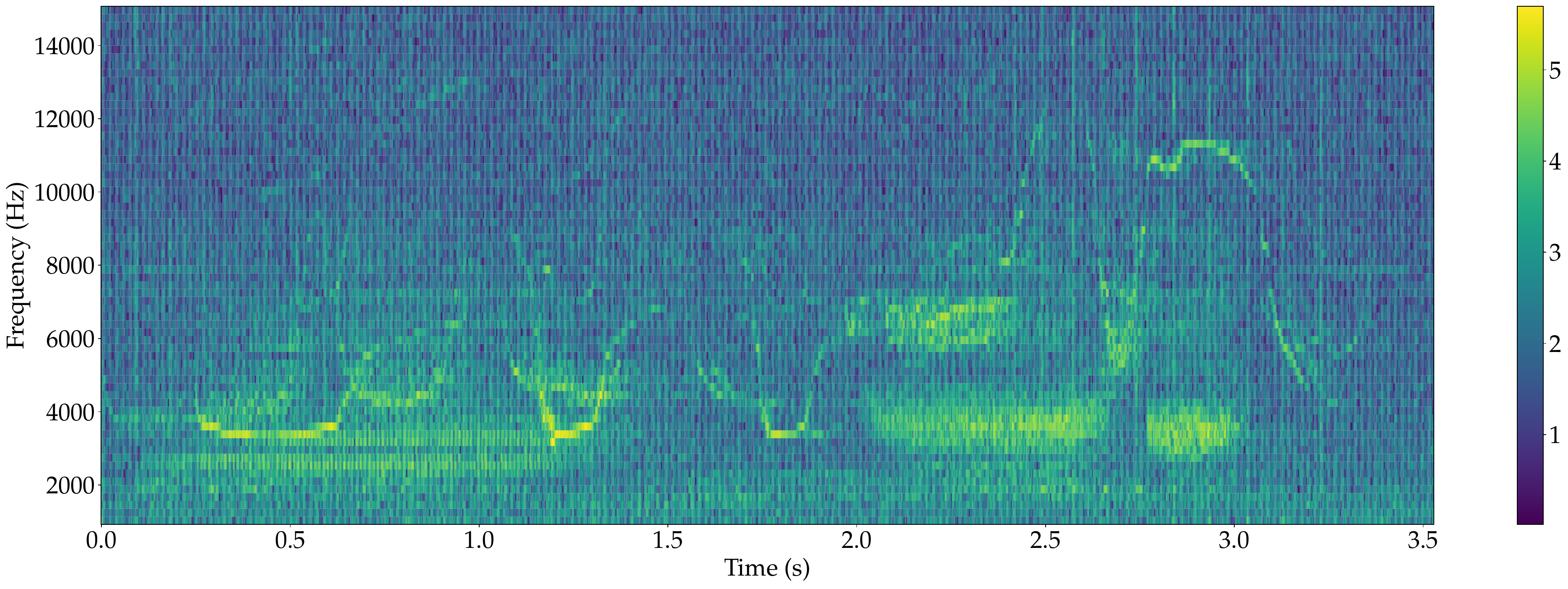}
    \caption{Spectrogram of the recording of a vocalizing dolphin.}
    \label{fig:Sz.audio}
\end{figure}

We set $\lambda=0.1$ and $\mu=0.00001$. The algorithm converges in $165$ iterations. The estimated spectrograms, shown in Figure~\ref{fig:results.dolphin},  illustrate the ability of our technique to isolate the two components in the time-frequency plane. On the one hand, the estimated spectrogram of the bumps component is a frequency-spread representation, especially visible during click periods (e.g., in the 2--3~s interval). On the other hand, the estimated spectrogram of the AM-FM signal captures the frequency-localized whistles (visible, e.g., around 4-6~kHz in the 0--2~s interval). This example shows the extent to which this tool can be used in bioacoustics to better visualize and interpret the two types of sound emitted simultaneously by a dolphin.

\begin{figure}
    \centering
    \includegraphics[width=.49\columnwidth]{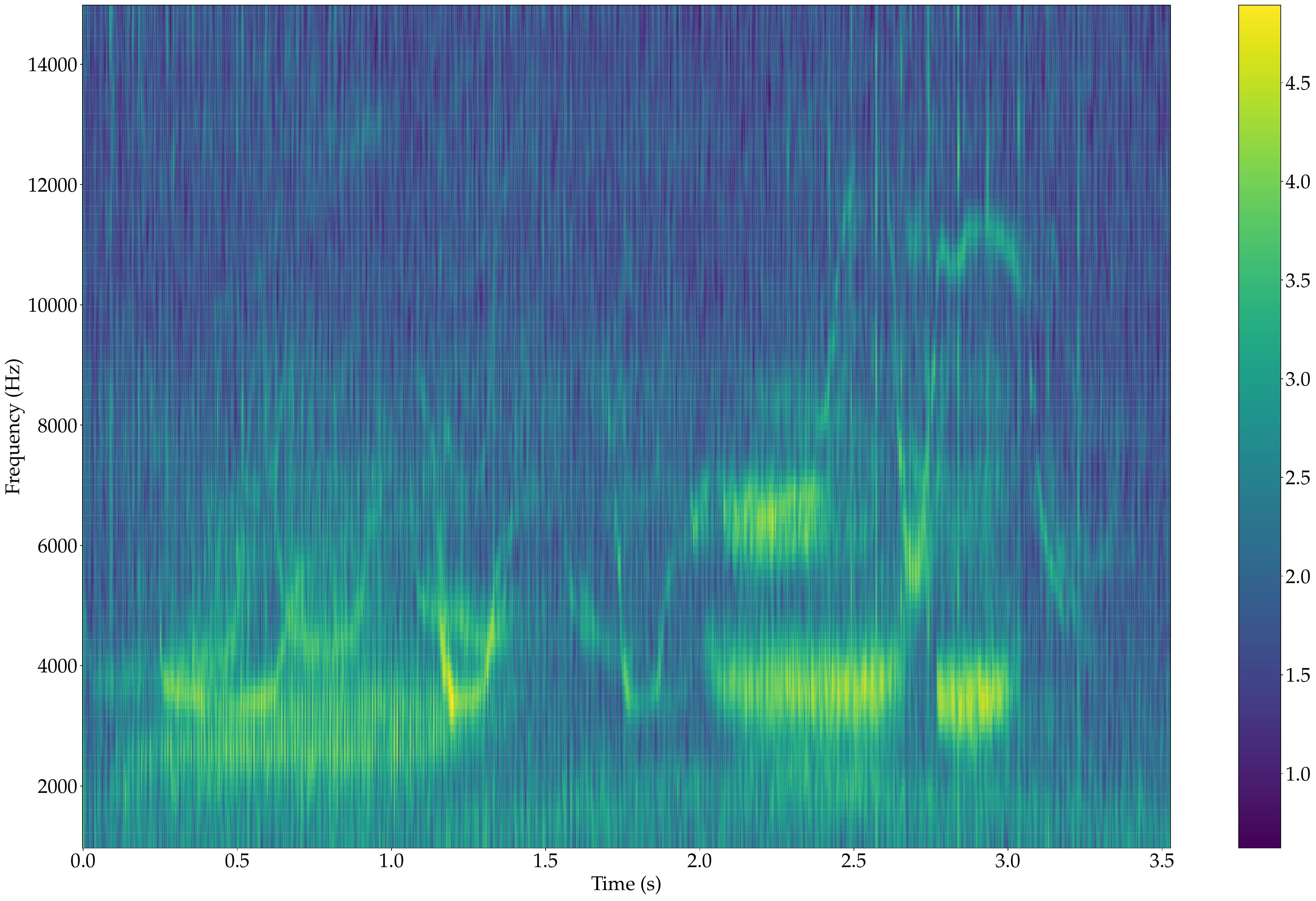}
    \includegraphics[width=.49\columnwidth]{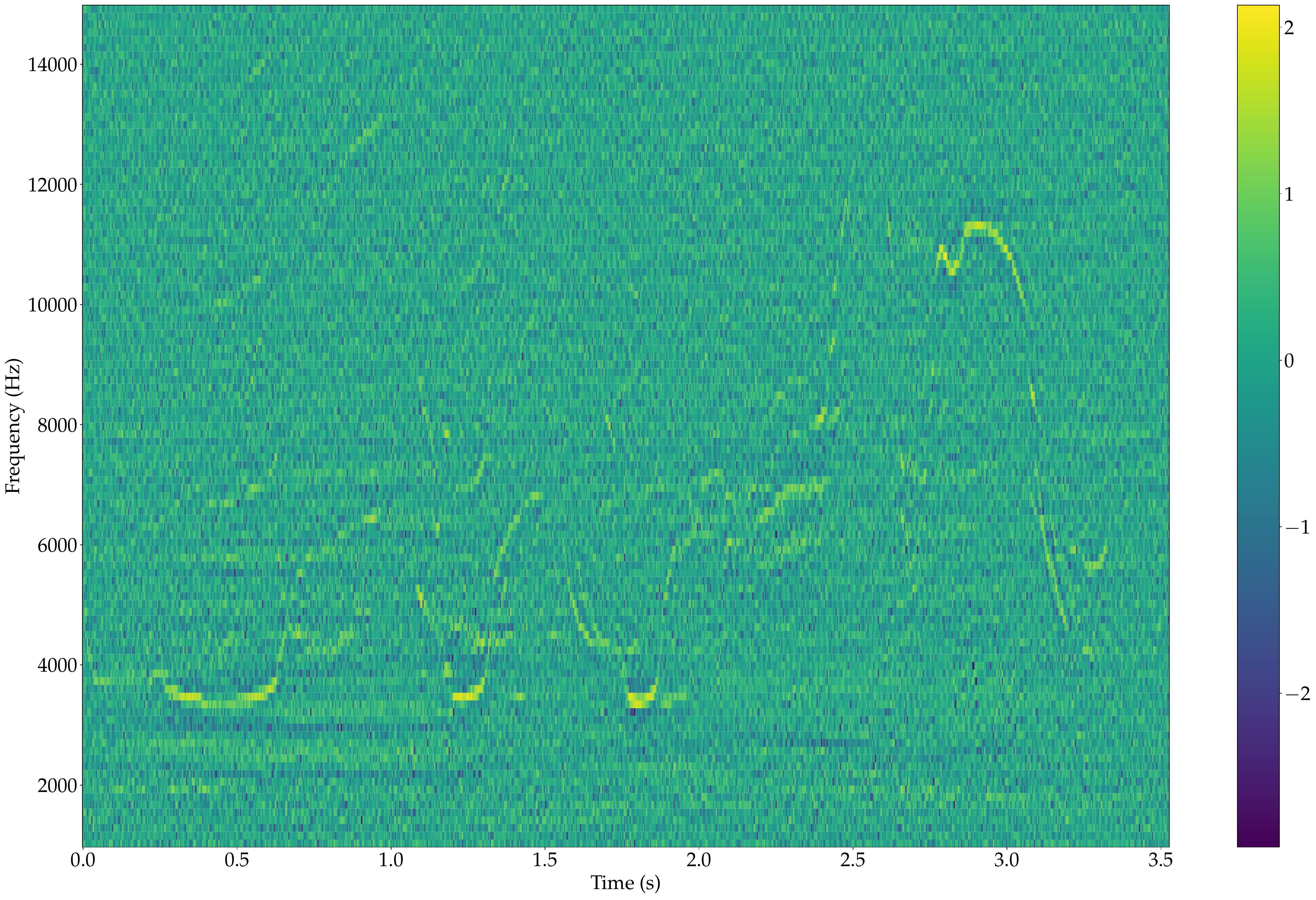}
    \caption{Estimated spectrograms of the clicks (left) and the whistles (right) of a vocalizing dolphin.}
    \label{fig:results.dolphin}
\end{figure}

\section{Conclusion and Perspectives}
\label{sec:concl}
We have presented a consistent spectrogram separation technique, adapted to a specific type of nonstationary mixtures composed of a bump component and an AM-FM component. To address this problem, we leveraged the properties of the time-frequency representations of such mixtures. We then constructed an inverse problem to estimate the spectrograms of the two components. Finally, we propose an alternating optimization algorithm that we test on a synthetic mixture. We show that our method outperforms the standard NMF decomposition. An application to a bioacoustic signal is also presented, demonstrating the real-world applicability of the technique. 

This work represents a step towards the construction of a single-channel source separation method for reconstructing the two components themselves. We plan to use the estimated spectrograms as time-frequency masks to estimate the STFTs of each component. Our previous work is particularly well-suited to this type of problem~\cite{Kreme2021}. An inverse transform will then be used to reconstruct the time signals.

\section*{Appendix: Proof of Theorem~\ref{th:approx.dSx}}
\label{se:proof}
\begin{proof}
By incorporating the bumps signal expression~\eqref{eq:model.x} into the STFT definition~\eqref{eq:stft}, we obtain
\begin{align*}
T_x(\nu,\tau) &= \sum_{k=1}^K\int_{\RR} \varphi(t-t_k) g(t-\tau) e^{-\ii 2\pi\nu t} \dd t \\
&= \sum_{k=1}^K \int_{\RR} \varphi(u) g(u-(\tau-t_k)) e^{-\ii 2\pi\nu (u-t_k)} \dd u \\
&= \sum_{k=1}^K e^{+\ii 2\pi\nu t_k} T_\varphi(\nu,\tau-t_k) .
\end{align*}
The second row is obtained by changing the variable $u=t-t_k$. Besides, since $\varphi$ and $g$ are compactly supported functions, $T_\varphi(\nu,\tau-t_k)$ is nonzero only if the supports of $\varphi$ and $g(\cdot-\tau-t_k)$ intersect. This condition is written
\begin{equation}
t_k-(\Delta_g+\Delta_\varphi)<\tau< t_k + \Delta_g+\Delta_\varphi .
\label{eq:condition.tau}
\end{equation}
Hence, $T_\varphi(\nu,\tau-t_k)$ and $T_\varphi(\nu,\tau-t_{k'})$ are never simultaneously nonzero if there is no value of $\tau$ where~\eqref{eq:condition.tau} is satisfied for $t_k$ and $t_{k'}$ simultaneously. We end up with the condition
\[
|t_k-t_{k'}|<2(\Delta_g+\Delta_\varphi),
\]
which is always satisfied when~\eqref{eq:condition.Delta} is true. Therefore, the spectrogram is given by
\begin{equation*}
S_x(\nu,\tau) =\! \left| \sum_{k=1}^K e^{\ii 2\pi\nu t_k} T_\varphi(\nu,\tau-t_k) \right|^2 \!=\! \sum_{k=1}^K \left| T_\varphi(\nu,\tau-t_k) \right|^2.
\end{equation*}
This yields
\begin{equation}
\partial_\nu S_x(\nu,\tau) = 2\Real\!\left( \sum_{k=1}^K T_\varphi(\nu,\tau\!-\!t_k)\overline{\partial_\nu T_\varphi(\nu,\tau\!-\!t_k)} \right).
\label{eq:Sx}
\end{equation}
Besides, the analysis window had wider support than $\varphi$. That is why we take advantage of the mean value theorem to write:
\begin{equation}
\exists 0<\vartheta(t)<t,\quad g(t-\tau) = g(-\tau)+t g'(\vartheta(t)-\tau).
\end{equation}
Hence,
\begin{align}
\nonumber
T_\varphi(\nu,\tau) &= \int_{-\Delta_\varphi}^{\Delta_\varphi}\varphi(t) g(t-\tau) e^{-\ii 2\pi\nu t} \dd t \\
&= g(-\tau)\hat\varphi(\nu) + \epsilon_1(\nu,\tau),
\label{eq:Sphi}
\end{align}
where $\epsilon_1(\nu,\tau-t_k)=\int_{-\Delta_\varphi}^{\Delta_\varphi}\varphi(t) t g'(\vartheta(t)-\tau) e^{-\ii 2\pi\nu t} \dd t$.
The same reasoning yields
\begin{equation}
\partial_\nu T_\varphi(\nu,\tau) = g(-\tau)\hat\varphi'(\nu) + \epsilon_2(\nu,\tau-t_k),
\label{eq:dnu.Sphi}
\end{equation}
where $\epsilon_2(\nu,\tau)=-\ii 2\pi\int_{-\Delta_\varphi}^{\Delta_\varphi}\varphi(t) t^2 g'(\vartheta(t)-\tau) e^{-\ii 2\pi\nu t} \dd t$.
Inserting results~\eqref{eq:Sphi} and~\eqref{eq:dnu.Sphi} into Equation~\eqref{eq:Sx} leads to the main result~\eqref{eq:dSx} of the theorem. Finally, we bound the error terms as follows:
\begin{align*}
|\epsilon_1(\nu,\tau)|&\leq \int_{-\Delta_\varphi}^{\Delta_\varphi}\|\varphi\|_\infty |t| \|g'\|_\infty \dd t = \Delta^2_\varphi\|\varphi\|_\infty\|g'\|_\infty \\
|\epsilon_2(\nu,\tau)|&\leq 2\pi\!\!\int_{-\Delta_\varphi}^{\Delta_\varphi}\!\!\|\varphi\|_\infty |t|^2 \|g'\|_\infty \dd t =  \frac43\pi\Delta^3_\varphi\|\varphi\|_\infty\|g'\|_\infty.
\end{align*}
\end{proof}

\bibliographystyle{ieeetr}
\bibliography{ref_eusipco24}

\end{document}